\documentclass[11pt]{article}
\usepackage{rotating,multirow,amsmath,amssymb,fullpage}
\usepackage{graphicx}
\usepackage{subfig}
\usepackage{latexsym,amsfonts,amsmath, amsthm, amssymb}
\usepackage{color}
\usepackage{tikz}
\usepackage{pgfplots}

\newtheorem{definition}{Definition}
\newtheorem{theorem}{Theorem}
\newtheorem{lemma}{Lemma}
\newtheorem{proposition}{Proposition}
\newtheorem{corollary}{Corollary}

\begin{document}
\title{Interruptible Algorithms for Multiproblem Solving}

\author{%
Spyros Angelopoulos%
\thanks{Sorbonne Universit\'e, CNRS, Laboratoire d'informatique de Paris 6, LIP6, F-75252 Paris, France. {\tt spyros.angelopoulos@lip6.fr}}
\and
Alejandro L\'opez-Ortiz
}

%
%

%
%
%

\newcommand{\braket}[2]{\langle #1,#2\rangle}

\newcommand{\far}{fair acceleration ratio}
\newcommand{\REMOVED}[1]{ }
\newcommand{\defi}{\textrm{def}}
\newcommand{\expo}{{\sc Exp}}
\newcommand{\perf}{\textrm{perf}}

\maketitle

\begin{abstract} 

In this paper we address the problem of designing an interruptible system 
in a setting in which $n$ problem instances, all equally important,
must be solved concurrently. The system involves scheduling executions of contract algorithms
(which offer a trade-off between allowable computation time and quality of the solution)
in $m$ identical parallel processors. When an interruption occurs, the system must 
report a solution to each of the $n$ problem instances.
The quality of this output is then compared to the best-possible algorithm that 
has foreknowledge of the interruption time and must, likewise, 
produce solutions to all $n$ problem instances.
This extends the well-studied setting in which only one problem instance 
is queried at interruption time.

In this work we first introduce new measures for evaluating the performance of interruptible systems in this setting. 
In particular, we propose the {\em deficiency} of a schedule as a performance measure that meets the requirements 
of the problem at hand. We then present a schedule whose performance we prove that is within a small factor from optimal in the 
general, multiprocessor setting. We also show several lower bounds on the deficiency of schedules on a single processor. 
More precisely, we prove a general lower bound of $(n+1)/n$, an improved lower bound for the two-problem setting  ($n=2$), and a tight lower
bound for the class of round-robin schedules. Our techniques can also yield a simpler, alternative proof of the main result of
Bernstein {\em et al.}~\cite{BFZ} concerning the performance of cyclic schedules in multiprocessor environments. 
\bigskip

\noindent
{\bf Keywords}: Anytime computation; contract algorithms; interruptible algorithms; acceleration ratio;
scheduling problems in Artificial Intelligence; performance measures in scheduling.

\end{abstract}

\section{Introduction}
\label{sec:introduction}

A designer of real-time systems should anticipate the situation in which there are limitations on the 
available execution time. Applications such as medical diagnosis systems, automated trading
systems and game-playing programs require that the system may be queried at any time during 
its execution, at which point a solution must be reported. 
{\em Anytime algorithms} occur precisely in such settings, namely in situations in which a computationally
difficult problem is addressed under uncertain running time availability. 
Such algorithms will produce an output whose quality improves as a function of the 
available computation time. 
Anytime algorithms were introduced by Horvitz~\cite{Hor.1987.reasoning},~\cite{Hor.1998.reasoning} 
and Dean and Boddy~\cite{DeaBod.1998.analysis} and arise in central AI problems such as
heuristic search, and planning under uncertainty~\cite{Zaimag96}.

Russel and Zilberstein~\cite{RZ.1991.composing} distinguish between two main classes 
of anytime algorithms. On the one hand, the
class of {\em interruptible algorithms} consists of algorithms
that can be interrupted at any point during their execution, and must always 
report their current (albeit not necessarily optimal) solution.  
On the other hand, the class of {\em contract algorithms} consists of algorithms 
which specify the exact amount of allowable computation time as part of their input.
Such algorithms must terminate their execution before a solution can be produced,
otherwise the output may be meaningless. 

Interruptible algorithms offer, by their definition, more flexibility; in contrast, contract algorithms are typically
much easier to design, implement and maintain~\cite{BPZF.2002.scheduling}. Thus, it is desirable, in general,
to be able to transform a contract algorithm to its interruptible variant. This can be addressed in an algorithm-specific
manner, but we are interested, instead, in ``black-box" techniques that are algorithm-independent.
Towards this end, a well-studied technique is by scheduling executions of contract algorithms,  
in a machine that consists either of a single or even multiple processors. 

More precisely, in the most general setting,  
we are presented with a set $P$ of $n$ problem instances which we want to solve, and for each problem we are 
given a contract algorithm for the said problem. In addition, 
we are given a system of $m$ identical processors on which we can schedule this sequence of contract algorithms. 
The goal is to devise an efficient schedule, that is a strategy that assigns interleaved executions 
of all contract algorithms on the processors.
In the standard setting, upon an interruption, a query for a problem in $P$ is issued. 
The algorithm must then report the solution of the (completed) contract algorithm with the longest execution time
for the queried problem, since the latter provides the best completed solution by the time of interruption.
Thus, we would like these ``lengths'' or durations of completed contracts
to be as large as possible, since the longer the execution time, the
better the quality of the solution returned by the contract algorithm (and thus by the interruptible
system as well). 

The standard performance measure for a schedule of contract algorithms is the  
{\em acceleration ratio}~\cite{RZ.1991.composing}. Informally, the measure describes the multiplicative increase in processor speed required
for the schedule to compensate for the lack of knowledge of the interruption time. 
More formally, let $l_{p,t}$ denote the length of the largest contract 
for problem $p$ completed by time $t$ in a schedule $X$. The acceleration ratio $\alpha(X)$ of schedule $X$ is then defined as
\begin{equation}
\alpha(X)= \sup_t \max_{p \in P} \frac{t}{l_{p,t}}. 
\label{eq:acc.ratio.definition}
\end{equation}

Thus, the accelration ratio is a worst-case measure that compares the quality of the solution returned 
by the schedule (that is, the quantity $l_{p,t}$) to an ideal, optimal algorithm that knows the interruption 
$t$ in advance and  dedicates a single processor in order to run a contract of length $t$ for problem $p$. 
It is worth noting that for $n=m=1$, i.e., for the setting of a single problem and a single processor, 
the problem of devising a schedule of minimum acceleration ratio is known in the context of online computation 
as the {\em online bidding} problem~\cite{ChrKen06} (though, to our knowledge, previous work has not identified this equivalence).

\paragraph*{Related work} \ 
Simulating interruptible algorithms by means of schedules of contract algorithms has been 
a topic of extensive study. 
Russell and Zilberstein~\cite{RZ.1991.composing} were the first to present such an explicit simulation. 
For the case of one problem instance and 
a single processor, they provided a schedule based on iterative doubling of contract lengths
for which the corresponding interruptible algorithm has acceleration ratio at most
four. Zilberstein {\em et al.}~\cite{ZilbersteinCC03} showed that in this
case the schedule is optimal, in the sense that no other schedule of better
acceleration ratio exists.

Zilberstein {\em et al.}~\cite{ZilbersteinCC03} addressed the generalization
of multiple problem instances (assuming a single available processor), and
Bernstein {\em et al.}~\cite{BPZF.2002.scheduling}  studied the generalization in which 
contracts for a single problem instance must be
scheduled on a set of multiple processors. 
For both cases, optimal schedules are derived.
Bernstein {\em et al.}~\cite{BFZ} addressed the problem in its full generality, namely the setting in which
$n$ problem instances are given and the schedule is implemented on $m$ processors.
In particular, they showed an upper bound of $\frac{n}{m} (\frac{m+n}{n})^{\frac{m+n}{m}}$ on the 
acceleration ratio; in addition, they
showed that the schedule is optimal for the class of {\em cyclic} schedules.
The latter is a somewhat restricted, but still very rich and intuitive class of schedules with round-robin characteristics.
This restriction was removed by L\'{o}pez-Ortiz
{\em et al.}~\cite{aaai06:contracts}, who showed that this acceleration ratio is indeed optimal
among all possible schedules. Angelopoulos {\em et al.}~\cite{soft-contracts} studied the
setting in which the interruptions are not absolute deadlines,
but instead an additional window of computational time may be provided.

The problem of devising schedules of optimal acceleration ratio has interesting parallels with another
well studied problem in both AI and Operations Research, namely the problem of searching for a hidden 
target in an environment that consists of a number of unbounded, concurrent lines: this is known as 
the {\em star search} or {\em ray search} problem. Bernstein {\em et al.}~\cite{BFZ} were the first to establish
connections between the two problems; more recently,~\cite{spyros:rays} 
explored further connections between these classes of problems 
in probabilistic, fault-tolerant and randomized settings. In both works, the acceleration ratio of the 
scheduling strategy is compared to the {\em competitive ratio}, which is the standard performance 
measure of search strategies.


\paragraph*{Contribution} \ 
The central observation that motivates our work is that 
the acceleration ratio becomes problematic, as a performance measure, when at interruption time the 
algorithm is required to return a solution to {\em all $n$ problems in $P$} instead of only 
to a specific queried problem. This arises, for instance, in systems which involve parallel executions
of different heuristics (and at interruption time, the best heuristic is chosen). Another example is
a medical diagnostic system which must perform concurrent evaluations for a number of medical issues.
Here, the decision of the expert may very well have to take into account all such evaluations.

Suppose that we are given a set of $n$ problems (indexed $0,\ldots ,n-1)$ and a set $M$ of $m$ identical processors
(indexed $0,\ldots ,m-1$).
Suppose, in particular, that $n>m$, 
and for the purposes of illustrating our argument, say that $n \gg m$. 
Note that it is not feasible for any ideal, optimal algorithm 
(that is, an algorithm with advance knowledge of the interruption time $t$)
to schedule $n$ contracts of length $t$ to $m$ processors, since $n>m$. 
In a sense, if we applied the acceleration ratio to this domain, 
we would compare the performance of an interruptible algorithm which is expected 
to make progress on all $n$ problems to an algorithm which only makes optimal progress on at most $m$:
such a comparison is rather not fair.
This shortcoming was noticed by Zilberstein {\em et al}~\cite{ZilbersteinCC03}, who
define a scaled-down variant of the acceleration ratio for the case of $n$ problems and one processor as the quantity 
$ \sup_t \max_{p \in P} \frac{t/n}{l_{p,t}}$. This measure describes, informally, an even distribution
of the processor time among the $n$ problem instances for the optimal (offline) schedule.

A different way of arguing about the above shortcoming is to consider the scenario in which 
at interruption time $t$ a single problem $p \in P$ is queried. Naturally, the interruptible algorithm
does not know neither $t$ or $p$ in advance, but then the optimal offline algorithm should be oblivious 
of $p$ as well, otherwise it would make progress only on $p$ while ignoring all other problems in $P$.
This suggests that the offline optimal algorithm implicit in the definition of the acceleration ratio  
is overly powerful, and better measures are needed for the setting that we study.


Motivated by the above observations, in this paper we address the problem of designing 
interruptible algorithms using schedules of executions of contract algorithms, 
assuming that $m$ identical processors are available, and $n$ problem instances, 
all equally significant, must be solved. 
We begin by considering 
measures alternative to the acceleration ratio (Section~\ref{sec:defs}), and we propose
the {\em deficiency} of a schedule as our measure of choice. 
In Section~\ref{sec:exponential} we present a schedule whose deficiency is very small 
and rapidly decreasing in the ratio $n/m$
(in contrast, the acceleration ratio of every schedule is known 
to approach infinity when $n/m \rightarrow \infty$). More precisely, we show analytically that its deficiency is at most
3.74 if $n\leq m$, and at most 4, if $n>m$. A numerical evaluation provides even smaller values. 

Even though the deficiency of the proposed schedule is small, we provide further theoretical justification for 
our choice of schedule. Namely, in Section~\ref{sec:lower} we present several lower bounds on the deficiency of 
schedules in the single-processor setting ($m=1$). More precisely, we prove a general lower bound of $(n+1)/n$, 
an improved bound for the two-problem setting ($n=2$), and a tight lower bound for round-robin schedules. We also
remark that the schedule is optimal for the setting $n=m=1$. The proofs are based on techniques originally developed 
in the context of search theory~\cite{gal03:rendezvous} that have also been applied to scheduling problems~\cite{aaai06:contracts,soft-contracts},
and which allow us to relate the performance of schedules with arbitrary contract lengths to that of schedules with exponentially increasing lengths. 
As a further illustration of the applicability of these techniques, we give a much simpler, alternative proof of the main result 
in~\cite{BFZ}, namely we identify the optimal acceleration ratio that can be achieved by cyclic schedules, in the multi-processor setting.

As a last remark for this section, it is worth noting that revisiting the definition of an ideal, optimal algorithm
(and introducing new measures that capture this weakening) is an often encountered concept. 
As an illustrative example, a multitude of measures alternative to the competitive ratio have been introduced in the context
of the analysis of online algorithms (see e.g., the survey~\cite{survey}). These measures were motivated by the observation that 
the concept of an offline optimal algorithm is often overly powerful, with the undesirable side-effect that many online 
algorithms are often deemed theoretically optimal, though very inefficient in practice. As a second example, we note that
for the problem of multi-target searching in a star,~\cite{oil,hyperbolic} consider a weakening of the optimal algorithm which knows 
the distance of the targets from the original position of the searcher, but not the precise ray on which each target is located. 
In both examples, the introduction of new measures gives rise to new algorithmic ideas and new techniques for analysis.

\section{Problem formulation and comparison of measures}
\label{sec:defs}

Consider an (infinite)  schedule $X$ of executions of contract algorithms
(also called {\em contracts}). For an interruption time $t$ we denote by
$l(X,p,t)$ the length of the longest contract 
for problem $p \in P$ which is finished by time $t$ in $X$ (or simply $l_{p,t}$ when $X$ is 
implied from context). We make the canonical assumption that at interruption time 
at least one contract per problem has already completed its execution.

We need to formalize the question: what is the best way to exploit the 
available resources (i.e., processors), in order to solve the set of problems $P$? Towards this end,
consider a schedule $Y$, which, in contrast to $X$, is finite: more precisely, $Y$ 
schedules $n$ distinct contracts, one for each problem in $P$ (if $Y$ schedules more
than one contract per problem, than we can transform $Y$ to a schedule $Y'$ which is at least
as good as $Y$ by keeping only the largest contract per problem that appears in $Y$). 
Each contract in $Y$ is scheduled in one of the $m$ processors in $M$. We require that $Y$
is {\em feasible with respect to $t$}, in the sense that the {\em makespan} of $Y$, namely the total
sum of contract lengths on the most loaded processor used by $Y$
does not exceed $t$. Let ${\cal Y}_t$ denote the
class of all schedules $Y$ with the above properties. We will be calling $Y$ an {\em offline solution}
since it relies on advance knowledge of $t$.

Having defined ${\cal Y}_t$, we need a measure of how a schedule $Y \in {\cal Y}_t$ compares 
to $X$, which will also dictate which is the {\em best} schedule in ${\cal Y}_t$ compared to $X$.
First observe that the acceleration ratio is not an
appropriate measure for our setting, since under it
the optimal offline schedule $Y$ for interruption $t$ dedicates all its resources to the contract 
that is worked on the least by $X$  while failing to produce an answer for all other problems. 
This results in a large acceleration ratio which however does not truly reflect the quality of $X$
(effectively, the optimal solution ``cheats'' by ignoring all but one problem instances, 
which is not acceptable in our setting).

Another alternative would be to compare the smallest contract completed by $X$ 
to the smallest contract completed by $Y$, by time $t$.
We will need some preliminary definitions first.
Let $S_X^t$ denote the set $\{ l(X, j,t)| j \in [0,n-1]\}$ (that is, the set of the largest contracts per problem 
completed by time $t$) and let $S_X^t(i)$ denote
the $i$-th smallest element of $S_X^t$. Similarly, for $Y \in {\cal Y}_t$ let $S_Y^t$ denote 
the set of $n$ contracts in $Y$, and $S_Y^t(i)$ be the $i$-th smallest contract in $Y$, respectively
(ties are resolved arbitrarily). 

Formally, we define the {\em performance ratio of $X$ with respect to $Y$ at time $t$} as:
\begin{equation}
\perf(X,Y,t)=\frac{\min_i S_Y^t(i)}{\min_i S_X^t(i)}= 
\frac{S_Y^t(1)}{S_X^t(1)}.
\end{equation}
The performance ratio of $X$ at time $t$ is then defined as
\[
\perf(X,t)=\sup_Y \perf(X,Y,t),
\]
where $Y$ is a feasible schedule wrt $t$. 
Last, the performance ratio of $X$ is
defined as $\perf(X)=\sup_t \perf(X,t)$.

The first observation is that under this measure there exists an ``optimal'' offline schedule
 in which all contracts have the same length: here, by ``optimal'' we mean 
 a schedule $Y \in {\cal Y}_t$ against which the performance ratio of $X$ is maximized. 
Indeed, given any offline schedule $Y$, consider any
schedule $Y'$ such that the length of all its contracts is $\min_i S_Y^t(i)$. 
Note that such a feasible $Y'$ exists, 
since all contracts in $Y$ are at least that long, and $Y$ itself is feasible. Then it follows
from the definition that $\perf(X,Y,t)=\perf(X,Y',t)$. 
We can show a close relationship of the performance ratio to the acceleration
ratio for the standard setting (namely, when we seek the solution only 
to the queried problem). 

\begin{lemma} There is a schedule $X$ such that for any arbitrary interruption time $t$
\[ \perf(X,t) = \left\{\begin{array}{ll}
               \frac{n}{m}\left(\frac{m+n}{n}\right)^{\frac{m+n}{m}} 
                                \quad\quad\quad & \mbox{for $m\geq n$} \\
               \frac{n}{m}\frac{1}{\lceil n/m \rceil} 
               \left(\frac{m+n}{n}\right)^{\frac{m+n}{m}}
                               & \mbox{for  $m < n$.}
                  \end{array} \right.
\]
Furthermore, $X$ is optimal with respect to this measure.
\label{lemma:performance.ratio}
\end{lemma}
\begin{proof}
Consider first the case 
$m\geq n$,  then an optimal offline schedule 
consists of executing one contract of length $t$ per problem, each on its own processor
and hence $ \perf(X,t)= \frac{t}{\min_i S_X(i)}$. 
Note then that $\perf(X)=\sup_t \perf(X,t)$ 
is precisely the definition of the acceleration ratio $\alpha(X)$, for which the results
of~\cite{BFZ} and~\cite{aaai06:contracts} show the optimal value
of $\frac{n}{m}\left(\frac{m+n}{n}\right)^{\frac{m+n}{n}} $.

Next, suppose that $m < n$, then every offline schedule $Y$ is such that there exists at least
one processor in which at least $\lceil n/m \rceil$ contracts are scheduled.
As argued earlier, there exists an optimal offline schedule in which all contracts have the same
length.
It follows then that there is an optimal offline schedule with contract lengths equal to $t/\lceil n /m\rceil$.
Therefore,
$ \perf(X,t)= \frac{t/\lceil n /m\rceil}{\min_i S_X(i)}$
and $\perf(X)=\sup_t \perf(X,t)$ can then be described as
$\frac{1}{\lceil n/m\rceil} \cdot \alpha(X)$,
and the lemma follows. 
\end{proof}

We note that for the case of one processor ($m=1$) and $n$ problems, Lemma~\ref{lemma:performance.ratio}
shows that the performance ratio of the optimal schedule matches the measure 
proposed by~\cite{ZilbersteinCC03}, as mentioned in Section~\ref{sec:introduction}.

It is also not difficult to show that, unlike the optimal acceleration ratio (which is bounded by a function linear in $\frac{n}{m}$),
the optimal performance ratio is bounded by a small constant; in other words, the schedule in the proof of Lemma~\ref{lemma:performance.ratio} 
is very efficient under this refined measure. 
\begin{proposition}
The performance ratio of the schedule in the proof of Lemma~\ref{lemma:performance.ratio} is bounded
by $2e$, if $n>m$, and by 4 , if $n \leq m$. 
\label{lemma:performance.bounded}
\end{proposition}
\begin{proof}
In the case $m \geq n$, we have 
\begin{eqnarray}
\perf(X)\! 
= \left(1+\frac{n}{m} \right ) \left(1+\frac{m}{n}\right)^{\frac{n}{m}}  \leq 2 \cdot 2=4. \nonumber
\end{eqnarray}

In the case $m<n$ we have that 
\begin{eqnarray}
\perf(X) \leq
\left(1+\frac{m}{n} \right ) \left(1+\frac{m}{n}\right)^{\frac{n}{m}} \leq 2 \cdot e=2e. \nonumber
\end{eqnarray}
In either case, $\perf(X) \leq 2e$.
\end{proof}

Figure~\ref{fig:perf} shows the plot of $\perf(X)$ as a function of the ratio $\frac{n}{m}$ for $n>m$,
and assuming, for simplicity, that $m$ divides $n$. 
Note how the performance of the schedule is between 4 and $e$, and decreases rapidly as $n/m$ increases. 
When $m \geq n$, $\perf(X)$ decreases in a similar manner, as $m/n$ increases, and takes values between 4 and 1.

\begin{figure}[ht]
\begin{minipage}[b]{1\linewidth}
\centering
    \includegraphics[scale=0.25]{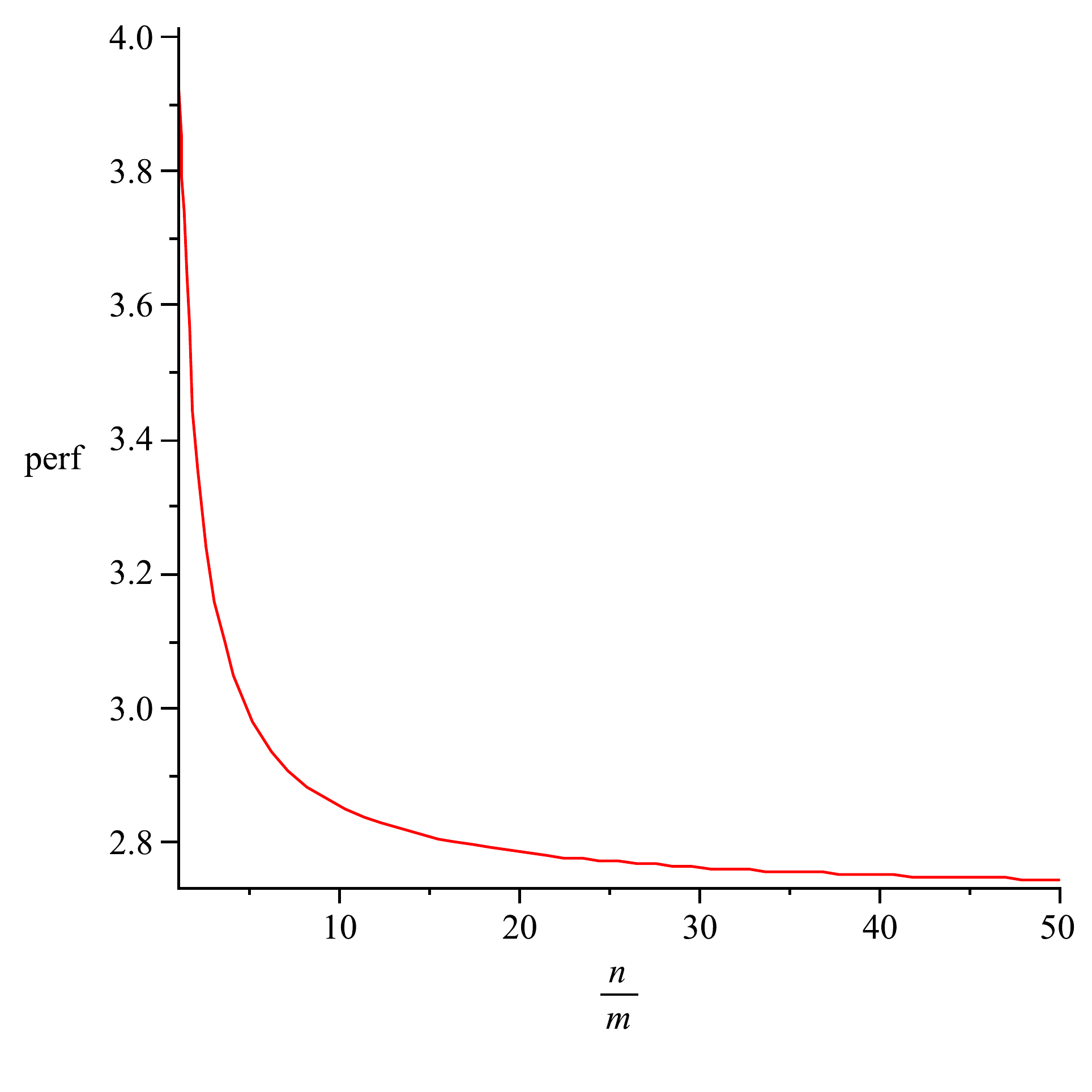}
\end{minipage}
\caption{\label{fig:perf} Plot of function $\perf(X)$, assuming $m$ divides $n$.}
\end{figure} 

While the performance ratio seems a better candidate for a measure in the context of our 
problem than the acceleration ratio, it is far from being the best possible choice. Note that according to this 
measure, every contract in the optimal offline solution may have fixed length, namely 
$t/\lceil n /m\rceil$. While it is guaranteed that the smallest contract in $S_X^t$ indeed
does not exceed $t/\lceil n /m\rceil$, the solution produced by $X$ may be such that there
exist several contracts in $S_X^t$ which exceed this length. More formally, there may exist
$j$ such that $S_X^t(j) > S_Y^t(j)$ for the optimal offline solution $Y$. This is clearly
undesirable, since it becomes difficult to argue that 
the optimal offline solution $Y$ at time $t$ is indeed better than the solution produced
by the interruptible algorithm at time $t$, even though, supposedly, $Y$ is optimal.

The above motivates the need for defining a further measure, one which takes into account the intuitive
expectation that the optimal offline solution should perform better than the 
interruptible algorithm {\em on each problem}. To accomplish this, we allow the offline
solution to observe the behavior of $X$ at each point in time $t$, and then produce
an appropriate ``optimal'' offline solution, tailored to the specifications of our problem.
In a sense we allow the offline solution sufficient power in choosing its schedule, while at the same
time we require that it produces solutions to all problem instances. This yields a measure 
which we call {\em deficiency}, and which
is: i) consistent with the requirements of the problem; and ii) powerful enough, in the sense
that if an algorithm performs well with respect to the new measure, then there are very strict
guarantees about its performance.

Formally, we say that $Y$ is {\em at least as good as $X$}, denoted by
$Y \geq X$, if and only if $S_Y^t(i) \geq S_X^t(i)$, for all $i \in [1,n]$. Then for a schedule
$Y \in {\cal Y}_t$ with $Y\geq X$, we say that the {\em deficiency of $X$ wrt $Y$} 
for interruption $t$ is defined as 
\begin{equation}
\defi(X,Y,t)=\min_i \frac{S_Y^t(i)}{S_X^t(i)}
\label{eq:deficiency}
\end{equation}
The {\em deficiency of $X$ given $t$} is then defined as 
\begin{equation}
\defi(X,t)=\sup_{Y\in {\cal Y}_t, Y\geq X} \defi(X,Y,t)
\label{eq:deficiency2}
\end{equation}
We define the deficiency of $X$ simply as
\begin{equation}
\defi(X)=\sup_t \ \defi(X,t)
\label{eq:deficiency3}
\end{equation}

While~\eqref{eq:deficiency} gives a formal definition of the deficiency, we will obtain an alternative
definition that provides us with more flexibility in terms of the analysis of actual schedules. 
To this end, we will relate this measure to the {\em makespan} of a schedule. 
\begin{definition}
Suppose that we are given $m$ identical processors and a set $S$ of $n$ jobs, where each job has a certain size (or length). 
For a given schedule of $S$ (i.e., for an assignment of $S$ to the processors), the {\em makespan} of the schedule is the load of the 
least loaded processor in this schedule, where the load of a processor 
is defined as the sum of the sizes of the jobs assigned to the said processor. We denote by $OPT(S)$ the optimal makespan of $S$,
among all possible schedules. 
\label{def:makespan}
\end{definition}
We will first show how to select a canonical 
representative from the set of schedules of optimal deficiency.
\begin{lemma}
Given schedule $X$ and interruption t, there exists an optimal schedule $Y$ with the following properties: 
\begin{enumerate}
\item $Y$ is such that $S_Y^t(i)=d\cdot S_X^t(i)$, for all $1\leq i \leq n$
and $d=\textup{def}(X,t)\geq 1$.
\item $d$ is the largest possible value such that $S_Y^t$ can be feasibly scheduled in $M$.
i.e., the makespan of the corresponding schedule is exactly $t$. 
\end{enumerate}
\label{lemma:basic}
\end{lemma}
\begin{proof}
Let $Y'$ be an optimal schedule given $X$ and $t$. It is not hard to transform $Y'$ to another optimal schedule $Y$
which satisfies property (1) as follows.  First, observe that from the definition of $\defi(X,t)$
it follows that $S_{Y'}^t(i) \geq d\cdot S_X(i)$. 
Now  for all $i$ such that  $S_{Y'}^t(i) >  d\cdot S_X^t(i)$, 
we can decrease the length of each contract $S_{Y'}^t(i)$ obtaining a new 
schedule $Y$ for which $S_Y^t(i) = d\cdot S_X^t(i)$. Clearly, $Y$ is a feasible optimal schedule as well, for the
given $X$.
For property (2) assume otherwise, i.e. that the makespan of $Y$ is not $t$. Since $Y$ is feasible with 
respect to $t$, by definition we have that the makespan 
of $Y$ must equal $t-\epsilon$ for some $\epsilon>0$. 
Now consider a new schedule $Y''$ identical to $Y$ except that
every contract length is multiplied by $t/(t-\epsilon)$. The makespan of $Y''$ is 
$t$, and thus the deficiency of $X$ is at least
\[
\defi(X,Y'',t) = \min_i \{S_{Y''}(i)/S_X(i)\} = dt/(t-\epsilon),
\]
which contradicts the fact that $\defi(X,t)=d$.
\end{proof}

Lemma~\ref{lemma:basic} implies the following corollary which establishes the relation between the deficiency and 
the makespan. 
\begin{corollary}
$\textup{def}(X,t)=\frac{t}{OPT(S_X^t)}$, where $OPT(S_X^t)$ is the minimum makespan for scheduling 
$n$ jobs in $m$ processors, with job $i$ having size (length) equal to $S_X^t(i)$.
\label{corollary:deficiency.formula}
\end{corollary}

Similarly to the acceleration ratio, one can think of the deficiency of a schedule as a resource-augmentation measure. Specifically,
it guarantees that if the contracts of $X$ are scheduled with a processor speedup of $\defi(X)$, then for any interruption $t$,
no schedule that knows $t$ can obtain a better feasible solution relatively to the one achieved by $X$.

\section{A near-optimal schedule for general $m$}
\label{sec:exponential}

In this section we propose a schedule for which we will show that the deficiency is bounded by a small constant, 
and is, thus, very close to optimal. 
More precisely, we will identify an efficient schedule that belongs in the class of {\em exponential} schedules. 
Informally, these are schedules in which problems are assigned to processors in round-robin order, and 
in which the lengths of the assigned contracts increases geometrically. 
\begin{definition}[\cite{BFZ}]
A schedule $X$ with contract lengths\footnote{With a slight abuse of notation, given a sequence of contract lengths $x_0,x_1, \ldots$, we will often refer to the contract of length $x_i$ as {\em the contract $x_i$}. This is only done for simplicity, and we do not require that contracts have pairwise different lengths.} $x_0,x_1, \ldots$ is called {\em round-robin} if it satisfies the following properties:
\begin{enumerate}
\item {\em Problem round-robin:} For every $i$, the contract $x_i$ is assigned to problem $i \bmod n$.
\item {\em Processor round-robin:} For every $i$, the contract $x_i$ is assigned to processor $i \bmod m$.
\end{enumerate}
\label{def:cyclic}
\end{definition}
\begin{definition}
A round-robin schedule $X$ with contract lengths $x_0,x_1, \ldots$ is called {\em exponential} if 
the contract length $x_i$ is equal to $b^i$, for some $b>1$.
\label{def:expo}
\end{definition}
Since an exponential schedule is fully described by its base, 
the remainder of this section is devoted to finding a value $b$ that yields a schedule of small  deficiency.  
Following~\cite{BFZ}, we will denote by
$G_k$ the finish time of the $k$-th contract of the schedule, in the round-robin order, whereas 
$L_k = b^k$ denotes the length of the $k$-th contract in this order.

The following lemma shows that in order to  evaluate the deficiency of any schedule (not necessarily
exponential) it suffices to consider only interruption times right before a contract terminates.
Let $G_c^-$ denote a time infinitesimally smaller than the finish time $G_c$ of contract $c \in X$. 
Recall also that according to Definition~\ref{def:makespan}, $OPT(S)$ 
denotes the optimal makespan for a schedule of a set $S$ of jobs in $m$ processors.  

\begin{lemma}
$\textup{def}(X) =\sup_{c \in X} \frac{G_C^-}{OPT(S_X^{G_c^-})}$. 
\label{lemma:def.exponential}
\end{lemma}
\begin{proof}
Let $t_1,t_2, \ldots $ be the sequence of finish times of contracts in $X$, in increasing order.
For any $t$ such that $t_i<t<t_{i+1}=G_c$ (for some contract $c$) we have that
$S_X^t=S_X^{G_c^-}$ (since no contract finishes in $(t_i,t_{i+1})$). 
From Corollary~\ref{corollary:deficiency.formula}  we have 
\[
\defi(X,t)=\frac{t}{OPT(S_X^t)}=\frac{t}{OPT(S_X^{G_c^-})} 
\leq \frac{G_C^-}{OPT(S_X^{G_c^-})}.
\] 
The lemma follows from $\defi(X)=\max_t \defi(X,t)$. 
\end{proof}

Since $X$ is a round-robin, exponential schedule, Lemma~\ref{lemma:def.exponential}
implies that
\begin{equation}
\defi(X)= \sup_{k\geq 0} \frac{G_{n+k}^-}{OPT(\{L_k, \ldots L_{n+k-1}\})}
\label{eq:def.expo}
\end{equation}
where $L_i=b^i$ is the length of the $i$-th contract in the round-robin order. 
Eq.~\eqref{eq:def.expo} suggests that in order to bound the deficiency of $X$, one needs to
obtain a lower bound on the makespan of $n$ jobs, whose lengths are equal to $L_k, \ldots L_{n+k-1}$,
respectively. This is accomplished in the next lemma. 
\begin{lemma}
For $L_i=b^i$ it holds that 
\[
OPT(L_k, \ldots L_{n+k-1}) \geq \kappa \cdot b^k \frac{b^{n+m-1}-b^{(n-1) \bmod{m}}}{b^m-1},
\]
where $\kappa$ is defined as $\kappa=\max \left \{ \frac{1}{2-1/m}, \frac{b^m-1}{b^{m}} \right \}$.
\label{lemma:graham.application}
\end{lemma}

\begin{proof}
We will evaluate the makespan of the well-known greedy algorithm 
which schedules jobs of sizes $L_k, \ldots L_{n+k-1}$, in $m$ identical processors 
numbered $0, \ldots m-1$, assuming the jobs are considered in this particular
order (i.e., in increasing order of size). More precisely, the algorithm will assign each job  
to the machine of least current load. It is easy to see that the decisions of the greedy algorithm
are such that 
the job of size $L_{k+i}$
is scheduled on processor $i \bmod{m}$. Moreover, the incurred makespan
is determined by the total load of jobs scheduled on the 
same processor as job $L_{n+k-1}$, namely processor $(n-1) \bmod{m}$. 
Denote by $Gr(L_k, \ldots L_{n+k-1})$ the makespan of the greedy algorithm. We obtain
\begin{eqnarray}
&Gr&(L_k, \ldots L_{n+k-1}) = \sum_{i=0}^{\lfloor (n-1)/m \rfloor} 
b^{k+mi+(n-1) \bmod{m}} \nonumber \\
&=& b^k b^{(n-1) \bmod{m}} \sum_{i=0}^{\lfloor (n-1)/m \rfloor} b^{mi} \nonumber \\
&=& b^k \frac{b^{n+m-1}-b^{(n-1) \bmod{m}}}{b^m-1}.  \label{eq:graham.application}
\end{eqnarray}
To complete the proof, we need to argue that the greedy scheduling policy does not achieve a makespan
worse that $(1/\kappa)$ times the optimal makespan. Graham's fundamental theorem on the performance of
the greedy scheduling policy~\cite{Graham66} states that the greedy algorithm has an approximation 
ratio of $2-1/m$. Moreover, we know that the optimal makespan is at least $b^{n+k-1}$, which in 
conjunction with~(\ref{eq:graham.application}) yields that the greedy algorithm is also a $b^m/(b^m-1)$
approximation (for our specific instance). The lemma follows by combining the above 
two approximation guarantees. 
\end{proof}

We now proceed to bound the deficiency of the exponential schedule $X$. It is easy to show 
that for exponential schedules,
$G_{n+k}=\frac{b^{k+n+m}-b^{(k+n) \bmod{m}}}{b^m-1}$~\cite{BFZ}. 
Let $\lambda=1/\kappa=\min \left \{ 2-\frac{1}{m}, \frac{b^m}{b^m-1} \right \}$. 
Combining with~(\ref{eq:def.expo}) and Lemma~\ref{lemma:graham.application} we obtain 
\begin{eqnarray}
\defi(X) &\leq&
\lambda  \cdot \sup_{k \geq 0} 
\frac{b^{k+n+m}-b^{(k+n) \bmod{m}}}
{b^k(b^{n+m-1}-b^{(n-1) \bmod{m}})} \nonumber \\ 
&=& \lambda \cdot \sup_{k\geq 0}\frac{b^{n+m}-(b^{(k+n) \bmod{m}})/b^k}
{b^{n+m-1}-b^{(n-1) \bmod{m}}} \nonumber \\
&\leq& 
\lambda \cdot \frac{b^{n+m}}{b^{n+m-1}-b^\gamma},
\label{eq:expo.deficiency}
\end{eqnarray}
where $\gamma$ is defined as $(n-1) \bmod{m}$. 

We thus seek the value of $b$ that minimizes the RHS of~(\ref{eq:expo.deficiency}).
It is not clear that this can be done analytically, since the factor $\lambda$ depends on $b$,
and the derivative of the RHS does not have roots that can be identified analytically. 
Instead, let $f(b)$ denote the function $b^{n+m}/(b^{n+m-1}-b^\gamma)$. 
Then $f(b)$ is minimized for a 
value of $b$ equal to $\beta= (n+m-\gamma)^\frac{1}{n+m+\gamma-1}$. Let $\rho$ be such that 
$n-1=\rho m+\gamma$, then $\beta=(m(\rho+1)+1)^\frac{1}{m(\rho+1)}$. 
Observe also that $f(b)=1/(b^{-1}-b^{\gamma-n-m})=1/(b^{-1}-b^{-m(\rho+1)-1})$. Summarizing, 
for $b=\beta$ we obtain a schedule of deficiency
\begin{equation} 
\defi(X)\leq \min \left \{2-\frac{1}{m}, \frac{\beta^m}{\beta^m-1} \right \} \cdot  
\frac{1}{{\beta}^{-1}-\beta^{-m(\rho+1)-1}}.
\label{eq:def.b}
\end{equation}

Figure~\ref{fig:exp} shows a plot of the deficiency of the schedule
(or, more accurately, the RHS of~(\ref{eq:def.b}) for the interesting case
where $n>m$) as function of $m$ and $\rho$.
Note that for fixed $m$, the deficiency increases as a function of $n$, until a 
point at which it becomes relatively stable with $n$ (this can be explained by the factor
$\lambda$ that is the minimum of two functions, one of which does not depend on $n$). 
For large $m$, and even larger $n$, the deficiency is close to 2. From the plot, the
maximum value of deficiency is $3/8 \cdot 5^{5/4} \approx 2.803$. 
\begin{figure}[htb!]
\begin{minipage}[b]{1\linewidth}
\centering
    \includegraphics[scale=0.35]{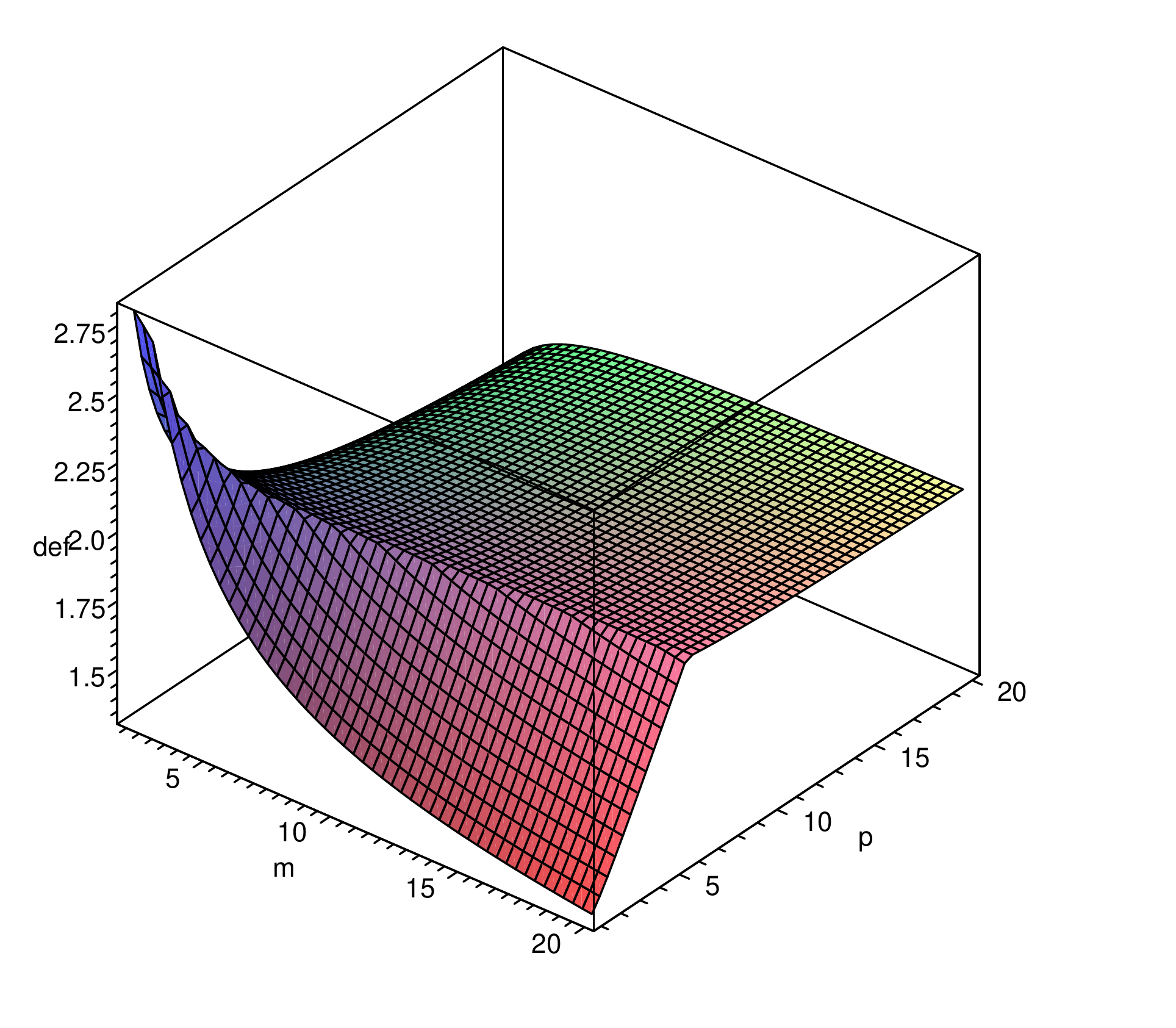}
\end{minipage}
\caption{\label{fig:exp} Plot of \defi(X) (as bounded by~\eqref{eq:def.b}) as a function of $m$ and $\rho$.}
\end{figure} 
Notwithstanding the numerical results, we can bound analytically the deficiency of the schedule $X$, as shown in the 
following proposition. 
\begin{proposition}
For the exponential schedule $X$ in which $b=\beta$, we have that $\textup{def(X)}\leq 3.74$, if $n>m$,
and $\textup{def(X)}\leq 4$, if $n \leq m$. 
\label{prop:def}
\end{proposition}
\begin{proof}
Substituting $m(\rho+1)$ with $y$ in~(\ref{eq:def.b}) we obtain 
\begin{equation}
\defi(X) \leq \left(2-\frac{1}{m}\right) \frac{1}{(y+1)^{-\frac{1}{y}}-(y+1)^{-\frac{y+1}{y}}}.
\label{eq:def.y}
\end{equation}
Using standard calculus it is straightforward to show that the functions $(y+1)^{-1/y}$  and
$(y+1)^{-\frac{y+1}{y}}$ are increasing and decreasing functions of $y$, respectively. 
Therefore, the denominator of the RHS of~(\ref{eq:def.y}) is an increasing function of $y=m(\rho+1)$. 
It turns out that for $n>m$ this upper bound on $\defi(X)$ is maximized when $m=2$ and $\rho=1$, 
hence $\defi(X)\leq 3.74$. When $n\leq m$, the $(2-1/m)$ factor vanishes, since the greedy schedule
achieves optimal makespan (in other words, $\lambda=1$). In this case $\defi(X)\leq 4$.
\end{proof}

We conclude this section by observing that the exponential schedule $X$ we obtained 
outperforms the strategy of optimal acceleration ratio, in all ranges of $n$ and $m$. 
More precisely, the schedule of~\cite{BFZ} and~\cite{aaai06:contracts} has unbounded
deficiency for $m\gg n$. 
For $n>m$, this schedule has constant deficiency, 
but larger than the deficiency of $X$ (namely, a deficiency equal to 4.24 in the worst case).

\section{Lower bounds on schedule deficiency for a single processor}
\label{sec:lower}
In this section we show several lower bounds on the deficiency of schedules, assuming $m=1$. In Section~\ref{subsec:lower.general} we show
a general lower bound on the deficiency of a schedule for $n$ problems. In Section~\ref{subsec:lower.roundrobin}
we give a matching lower bound for the class of round-robin schedules. The result shows that the exponential schedule 
of Section~\ref{sec:exponential} is optimal for this class of schedules. Last, in Section~\ref{subsec:lower.two} we give 
an improved lower bound on the deficiency of a schedule for two problems on a single processor. The proof demonstrates, in addition,
some of the difficulties that one has to bypass in order to improve the lower bound on arbitrary schedules, and will hopefully provide
some intuition about how to handle the general case. 

We first note that, for $m=1$, we obtain from~\eqref{eq:expo.deficiency} that
an exponential schedule $X$ with base $b$ has deficiency 
\begin{equation}
\defi(X) \leq \frac{b^{n+1}}{b^n-1},
\label{eq:def.expo.more}
\end{equation}
which is minimized for a value of $b$ equal to $(n+1)^\frac{1}{n}$. 
\begin{corollary}
The deficiency of the best exponential schedule for $n$ problems is at most 
$\frac{(n+1)^{\frac{n+1}{n}}}{n}$.
\label{cor:def.best.expo}
\end{corollary}
In what follows, we will denote by \expo \ the best exponential schedule whose deficiency is
described by Corollary~\ref{cor:def.best.expo}.

Recall that we denote by $l_{p,t}$ the length of the longest 
contract for problem $p$ that has completed by time $t$. We observe that for all schedules $X$ on a single processor,
\begin{equation}
\defi(X)=\sup_t \frac{t}{\sum_{i=0}^{n-1} l_{i,t}}=\sup_{t \in T_X} \frac{t}{\sum_{i=0}^{n-1} l_{i,t}},
\label{eq:def}
\end{equation}
where $T_X$ is defined as the set of all times right before the completion of a contract in $X$. 

\subsection{A general lower bound}
\label{subsec:lower.general}

\begin{theorem}
For any schedule $X$ that involves $n$ problems and a single processor, we have that
\[
\textup{def}(X) \geq\frac{n+1}{n}.
\]
\label{thm:lower.general}
\end{theorem}
\begin{proof}
Suppose, by way of contradiction, that $\defi(X)<\frac{n+1}{n}$. 
From~\eqref{eq:def}, for any given time $t$ we have 
\[
\frac{n+1}{n}>\defi(X) \geq \frac{t}{\sum_{i=0}^{n-1} l_{i,t}},
\]
therefore $\sum_{i=0}^{n-1} l_{i,t}> t\frac{n}{n+1}$. This implies there exists a problem for which its largest contract completed by time $t$ has length
at least $\frac{t}{n+1}$. Let $q\leq t$ denote the completion time of this contract. From the definition of the deficiency, and by the observation that $\sum_{i=0}^{n-1} l_{i,q^-} \leq q-\frac{t}{n+1}$, we obtain that
\[
\defi(X) \geq \frac{q}{q-\frac{t}{n+1}} \geq \frac{t}{t-\frac{t}{n+1}},
\]  
where the last inequality follows from the monotonicity of the function $f(x)=x/(x-a)$, with $a>0$. Therefore, $\defi(X) \geq\frac{n+1}{n}$, 
a contradiction.
 \end{proof}

Figure~\ref{fig:exponential-lowerbound} illustrates the performance of the best exponential schedule vs the lower-bound value
of Theorem~\ref{thm:lower.general}, for $n=1\ldots20$.
\begin{figure}[htb!]
\centering
\begin{tikzpicture}
    \begin{axis}[
        axis lines=center, xmin=1, xmax=20,
        ymax = 4,
        ylabel=\defi,
        xlabel=$n$, scale=0.7
        ]
        \addplot [domain=1:20,samples=250, thick, blue] {(x+1)/x};
        \addplot [domain=1:20,samples=250, thick, red ] {(x+1)^((x+1)/x)/x};
    \end{axis}
    \end{tikzpicture}
\caption{Plot of the deficiency of the best exponential schedule (in red) vs the the lower-bound value
of Theorem~\ref{thm:lower.general} (in blue).}
\label{fig:exponential-lowerbound}
\end{figure}
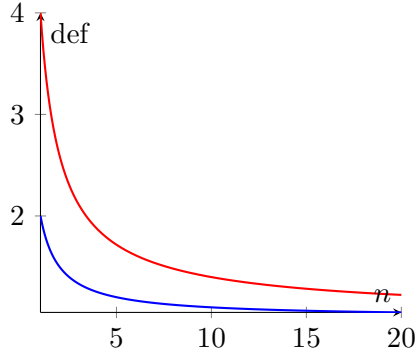

\subsection{Round-robin schedules on a single processor}
\label{subsec:lower.roundrobin}

We now consider the class of all round-robin schedules on a single processor. Namely, a schedule $X$ in this class schedules contracts of  
lengths $x_0,x_1, \ldots$, in this order, such that the contract $x_i$ is assigned to problem $i \bmod n$. We can assume, without
loss of generality, than for all $i \geq n$, $x_i<x_{i+n}$, otherwise the contract $x_{i+n}$ can be omitted from the schedule without 
increasing the deficiency of the schedule. 
We will show that the exponential schedule \expo \ is optimal for this class of
schedules. To this end, we will apply a theorem that allows us to relate the performance of schedules that have certain
structure (such as round-robin schedules) to the performance of exponential schedules.

More precisely, we will make use of the results by Gal \cite{gal80:search-games} and
Schuierer~\cite{schuierer:lb} which we state here in a simplified form. Given an infinite sequence 
$X=(x_0, x_1, \ldots)$, define $X^{+i}=(x_i,x_{i+1}, \ldots)$ as the suffix of the sequence $X$ starting at $x_i$. 
Define also $G_a = (1,a,a^2,\ldots)$ to be the geometric sequence in $a$.
\begin{theorem}[\cite{gal80:search-games,schuierer:lb}]\label{thm:limit}
  Let $X = (x_0,x_1,\ldots)$ be a sequence of positive numbers, $r$
  an integer, and $a = \limsup_{n\rightarrow\infty} (x_n)^{1/n}$, for
$a\in \mathbb{R} \cup
\{+\infty\}$.  
Let $F_k$, $k \geq 0$ be a sequence of functionals which satisfy the following properties:
  \begin{enumerate}
  \item $F_k(X)$ only depends on $x_0,x_1, \ldots ,x_{k+r}$,
    \item $F_k(X)$ is continuous for all $x_i>0$, with $0 \leq i
      \leq k + r$, 
    \item $F_k(\alpha X) = F_k(X)$, for all $\alpha > 0$, 
    \item $F_k(X+Y) \leq \max(F_k(X),F_k(Y))$, and
    \item $F_{k+1}(X) \geq F_k(X^{i+1})$, for all $i \geq
      1$,  
  \end{enumerate}
  then
  \[
     \sup_{0 \leq k < \infty} F_k(X)
     \geq  \sup_{0 \leq k < \infty} F_k(G_a).
  \]
\end{theorem}
At an intuitive level, Theorem~\ref{thm:limit} allows us to lower-bound the supremum of an infinite set of functionals by the supremum of
simple functionals of a geometric sequence; the latter is typically much easier to compute than the former. In our context, the objective
is then to express the deficiency as the supremum of a sequence of functionals. 
This technique has been applied in search games~\cite{gal03:rendezvous}, but also in previous work on 
scheduling of contract algorithms~\cite{aaai06:contracts,soft-contracts}.

\begin{theorem}
For every round robin schedule $X$ on a single processor, we have that $\textup{def(X)} \geq \frac{(n+1)^{\frac{n+1}{n}}}{n}$.
\label{thm:roundrobin}
\end{theorem}
\begin{proof}
Consider a time $t$ right before the completion of contract $x_{k+n}$ of the round-robin schedule, with $k\geq 0$. The set $S_X^t$ of the $n$ largest 
contracts per problem is then the set $\{x_k,x_{k+1},\ldots x_{k+n-1}\}$. Thus,
\[
\defi(X) \geq \frac{t}{\sum_{j=0}^{n-1} x_{k+j}}= \frac{\sum_{j=0}^{k+n}x_j}{\sum_{j=k}^{k+n-1} x_{j}}.
\]
Define now the functional $F_k(X)=\frac{\sum_{j=0}^{k+n}x_j}{\sum_{j=k}^{k+n-1} x_{j}}$. It is easy to verify that 
$F_k(X)$ satisfies the conditions of Theorem~\ref{thm:limit}; one can also appeal to Example 7.3 in~\cite{gal03:rendezvous}. Therefore, we conclude that there is $a>0$ such that 
\begin{equation}
\defi(X)=\sup_{0 \leq k < \infty} F_k(X)\geq \sup_{0 \leq k < \infty} \frac{\sum_{j=0}^{k+n}a^j}{\sum_{j=k}^{k+n-1} a^{j}}.
\label{eq:lower.round.1}
\end{equation}
Note that if $a=1$, we have that $\frac{\sum_{j=0}^{k+n}a^j}{\sum_{j=k}^{k+n-1} a^{j}}=\frac{k+n+1}{n}$, which tends to infinity as 
$k \rightarrow \infty$, and thus $\defi(X)=\infty$ in this case. Thus we can assume that $a\neq1$, thus~\eqref{eq:lower.round.1} implies
\begin{equation}
\defi(X)\geq \sup_{0 \leq k < \infty} \frac{a^{k+n+1}-1}{a^{k+n}-a^k}=\sup_{0 \leq k < \infty} \frac{a^{n+1}-\frac{1}{a^k}}{a^n-1}.
\label{eq:lower.round.2}
\end{equation}
Note that if $a<1$, then the RfHS of~\eqref{eq:lower.round.2} is $\infty$. Thus, we can assume that $a<1$. In this case, 
\begin{equation}
\sup_{0 \leq k < \infty} \frac{a^{n+1}-\frac{1}{a^k}}{a^n-1}=\frac{a^{n+1}}{a^n-1}.
\label{eq:lower.round.3}
\end{equation}
Last, the expression $\frac{a^{n+1}}{a^n-1}$ attains its minimum at $a=(n+1)^\frac{1}{n}$. For this value, and combining~\eqref{eq:lower.round.2}
and~\eqref{eq:lower.round.3} we obtain that 
\[
\defi(X) \geq \frac{(n+1)^{\frac{n+1}{n}}}{n},
\]
which concludes the proof.
\end{proof}

From Corollary~\ref{cor:def.best.expo}, it follows that \expo \ is optimal for the class of round-robin schedules. 

In order to illustrate the further applicability of Theorem~\ref{thm:limit}, we will show how it can lead to a simple alternative proof
of the main result in~\cite{BFZ}. Namely, we will show that there is an exponential schedule for scheduling contracts for $n$ problems in
$m$ parallel processors that  achieves optimal {\em acceleration ratio} for the class of {\em cyclic} schedules.
This class of schedules was introduced in~\cite{BFZ} and consists of round-robin schedules (i.e., schedules that satisfy the properties of
Definition~\ref{def:cyclic}), with the additional {\em length-increasing} property: namely, suppose that $i,j$ are such that the contracts $x_i, x_j$ are assigned to the same problem $p$. 
Then, if $i<j$, we require that $x_i<x_j$.

Benstein {\em et al.}~\cite{BFZ} gave an exponential cyclic schedule with acceleration ratio $\frac{n}{m}(\frac{n+m}{n})^\frac{n+m}{m}$.
Their main result showed that this schedule is optimal for the class of cyclic schedules:
\begin{theorem}[\cite{BFZ}]
For scheduling contracts for $n$ problems in $m$ processors, every cyclic schedule $X$ has acceleration ratio $\alpha(X)$ such that 
\[
\alpha(X) \geq \frac{n}{m}(\frac{n+m}{n})^\frac{n+m}{m}.
\]
\label{thm:steins}
\end{theorem}

\begin{proof}[Alternative proof of Theorem~\ref{thm:steins}] 
Let $X=(x_0,x_1,\ldots)$ denote a given cyclic schedule. For a given $k \geq 0$, let $p_k$ denote the processor to which contract $x_{k+n+m}$ is assigned,
and let the set $I_k$ denote the indices of contracts scheduled in $p_k$ and completed up to, and including the completion time of $x_{k+n+m}$. 
We have 
\begin{equation}
\alpha(X) \geq \max_{l\in [k,k+m-1]} \frac{\sum_{i \in I_l} x_i}{x_{l+m}}.
\label{eq:lower:acceleration.1}
\end{equation}
This is because the contract $x_{l+n+m}$ is completed at time $\sum_{i \in I_l} x_i$, and because right before this time, the longest contract
completed for problem $l$ has length $x_{l+m}$. The latter follows from the length-increasing and problem-round-robin properties
of the schedule.

Using the property $\max(a/c,b/d) \geq (a+b)/(c+d)$, for $a,b,c,d>0$,~\eqref{eq:lower:acceleration.1} gives
\begin{equation}
\alpha(X) \geq  \frac{\sum_{l=k}^{k+m-1}\sum_{i \in I_l} x_i}{\sum_{l=k}^{k+m-1}x_{l+m}}
\label{eq:lower:acceleration.2}
\end{equation}
Due to the processor-round-robin property of $X$, we have that 
\[
\sum_{l=k}^{k+m-1}\sum_{i \in I_l} x_i=\sum_{l=0}^{k+n+2m-1} x_l,
\]
and since~\eqref{eq:lower:acceleration.2} holds for all $k\geq 0$, we obtain that  
\[
\alpha(X) \geq  \sup_{0 \leq k < \infty} \frac{\sum_{l=0}^{k+n+2m-1} x_l}{\sum_{l=k+m}^{k+2m-1}x_{l}}.
\]
Define now the functional $F_k(X)=\frac{\sum_{l=0}^{k+n+2m-1} x_l}{\sum_{l=k+m}^{k+2m-1}x_{l}}$. Once again, it is easy to verify that 
$F_k(X)$ satisfies the conditions of Theorem~\ref{thm:limit}. Therefore, we conclude that there is $a>0$ such that 
\begin{equation}
\alpha(X) \geq \sup_{0 \leq k < \infty} \frac{\sum_{l=0}^{k+n+2m-1} a^l}{\sum_{l=k+m}^{k+2m-1}a^{l}}, \ \textrm{for some } a\geq 0.
\label{eq:lower:acceleration.4}
\end{equation}
Similar to the proof of Theorem~\ref{thm:roundrobin}, it is easy to see that if $a\leq1$, then the RHS of~\eqref{eq:lower:acceleration.4} is
$\infty$. We can thus assume that $a>1$, in which case we obtain that
\begin{eqnarray}
 \alpha(X) &\geq& \sup_{0 \leq k < \infty} \frac{a^{k+n+2m}-1}{a^{k+m}(a^m-1)} \nonumber \\
 &=& \sup_{0 \leq k < \infty}\frac{a^{n+m}-\frac{1}{a^{k+m}}}{a^m-1}=
 \frac{a^{n+m}}{a^m-1}.
 \label{eq:lower:acceleration.5}
 \end{eqnarray}
Last, the expression $\frac{a^{n+m}}{a^m-1}$ attains its minimum at $a=((m+n)/n)^{1/m}$. For this value of $a$,~\eqref{eq:lower:acceleration.5}
gives 
\[
\alpha(X) \geq \frac{n}{m}(\frac{n+m}{n})^\frac{n+m}{m},
\]
which concludes the proof.
\end{proof}

\subsection{Schedule deficiency for $n\in\{1,2\}$, $m=1$}
\label{subsec:lower.two}

In this section we will show improved lower bounds on the deficiency of schedules that involve two problems and a single processor. 
First, we argue that for $n=1$, the exponential schedule \expo \ is optimal. 
Second, we show how to obtain a lower bound on the deficiency of a schedule for $n=2$ that improves upon the lower bound of 
Theorem~\ref{thm:lower.general}. Namely, we obtain a lower bound of 2.115 whereas Theorem~\ref{thm:lower.general} gives a lower bound of only  1.5. 
Note that the corresponding upper bound due to \expo \ is 2.598. 

We can assume, without loss of generality, that for any given problem $p$, if an optimal schedule begins the execution
of two contracts for problem $p$ at times $t_1,t_2$, with $t_1<t_2$, then the length of the contract starting at time $t_1$ is
strictly smaller than the length of the contract starting at time $t_2$ (otherwise, the contract starting at time $t_2$ may be omitted 
altogether from the schedule, without increasing its deficiency).

\begin{proposition}
For $n=1$ and $m=1$, \expo \  has optimal deficiency.
\label{thm:exp.optimal.n.equals.1}
\end{proposition}
\begin{proof}
Let $X=(x_0,x_1,\ldots)$ denote a given schedule. From~\eqref{eq:def}, we have 
\[
\defi(X)=\sup_{t}\max_p \frac{t}{l_{p,t}}.
\]
Thus, the deficiency of $X$ is identical to its acceleration ratio $\alpha(X)$. 
From~\cite{RZ.1991.composing}, we know that the optimal
acceleration ratio equals 4, and is achieved by an exponential schedule which is precisely $X$.
\end{proof}

Next, we consider the setting $n=2$, $m=1$. We will first prove the existence of optimal schedules with a useful property.
Informally, the property states that there exists an optimal schedule such that any time a new contract is about to be scheduled,
the problem that has been worked the least will be chosen. We will call such optimal schedules {\em normalized}. The next lemma
proves this property for general $n$, assuming a single processor.

\begin{lemma}
For scheduling contracts for $n$ problems in a single processor, there exists a schedule of optimal deficiency 
which satisfies the following property: If at time $T$ a new contract
is about to be started, then it is assigned to a problem $i$ which minimizes $l_{i,T}$.
\label{lemma:normalization}
\end{lemma}
\begin{proof}
Let $X$ denote an optimal schedule, and suppose, by way of contradiction, that $T$ is the first time in which a contract
is about to begin its execution in $X$ that does not satisfy the property. More precisely, suppose that at time $T$, $X$ executes a contract
for a problem $j\neq i$ such that $l_{j,T}>l_{i,T}$. We will show that there exists a schedule $X'$ such that 
$X'$ is also an optimal schedule, and it satisfies the property for all $t \leq T$. By successively applying a series of such 
transformations, we obtain an optimal schedule that satisfies the property. 

We define a schedule $X'$ which is identical to $X$, with the exception that all contracts in $X$ which are assigned to problems 
$i$ and $j$ (defined as above) and which start their execution at time $t>T$ will switch problem assignment. In other words, a contract
that was assigned to problem $i$ will now be assigned to problem $j$, and vice versa. Let $T'$ be the first time such that $T'>T$
and a contract for problem $i$ is executed in $X$. Moreover, let $L$ and $L'$ denote the lengths of the two contracts for problems
$j$ and $i$ that begin their executions at times $T$ and $T'$, respectively, in $X$. From the definition of $X$
and~\eqref{eq:def}, we observe that for all $t$ such that $t< T+L$ or $t\geq T'+L'$, we have that $\defi(X',t)=\defi(X,t)$.

It thus remains to show that  $\defi(X',t)\leq \defi(X,t)$ for all $t\in [T+L,T'+L')$. For any such $t$, 
we have that
\[
\defi(X,t)=\frac{t}{\sum_{k \notin\{i,j\}} l_{k,t}+l_{j,t}+l_{i,T}},
\]
whereas 
\[
\defi(X',t)=\frac{t}{\sum_{k \notin\{i,j\}} l_{k,t}+l_{j,t}+l_{j,T}},
\]
Since $l_{i,T}<l_{j,T}$, it follows that $\defi(X',t)<\defi(X,t)$, for all $t\in [T+L,T'+L')$, which concludes the proof.
\end{proof}

\begin{lemma}
For $n=2$ and $m=1$ there is an optimal normalized schedule in which at most two contracts per problem are scheduled 
consecutively. 
\label{lemma:two.problems}
\end{lemma}

\begin{proof}
Let $X$ denote a normalized schedule. Let $t$ denote the earliest time in $X$ such that two consecutive contracts, say for problem $p=0$,
are about to be scheduled, and let $x_i,x_{i+1}$ denote their lengths respectively ($x_{i+1}>x_i$). Let also $l_0$, $l_1$, denote the 
largest contracts for problems $0,1$, respectively, that are completed in $X$ by time $t$, and $\lambda_1$ the largest contract for problem
1 completed by time $t^-$, i.e., infinitesimally smaller than $t$. Since $X$ is normalized, we have that $\lambda_1 \leq l_0 \leq l_1$, 
and $x_i\leq l_1$.

Consider the schedule $X'$ that is obtained by $X$ by omitting the contract $x_{i+1}$. Namely, all other contracts in $X'$ are scheduled in the 
same order, with no idle time (see Figure~\ref{fig:normalization} for an illustration). Note that $X'$ remains normalized. 
We will show that at least one of the following properties 
hold: i) $\defi(X') \leq \defi(X)$; or ii) $x_{i+1} \geq l_1$.
If the first property holds, this implies that we can effectively remove the first of the successively scheduled contracts without affecting
the optimality of the schedule. If the second property holds, since $X$ is normalized, this means that the contract to be scheduled right after 
$x_{i+1}$ in $X$ will be a contract for problem 1. The combination of these observations implies the lemma. 

\begin{figure}[htb!]
\centering
\includegraphics[scale=0.25]{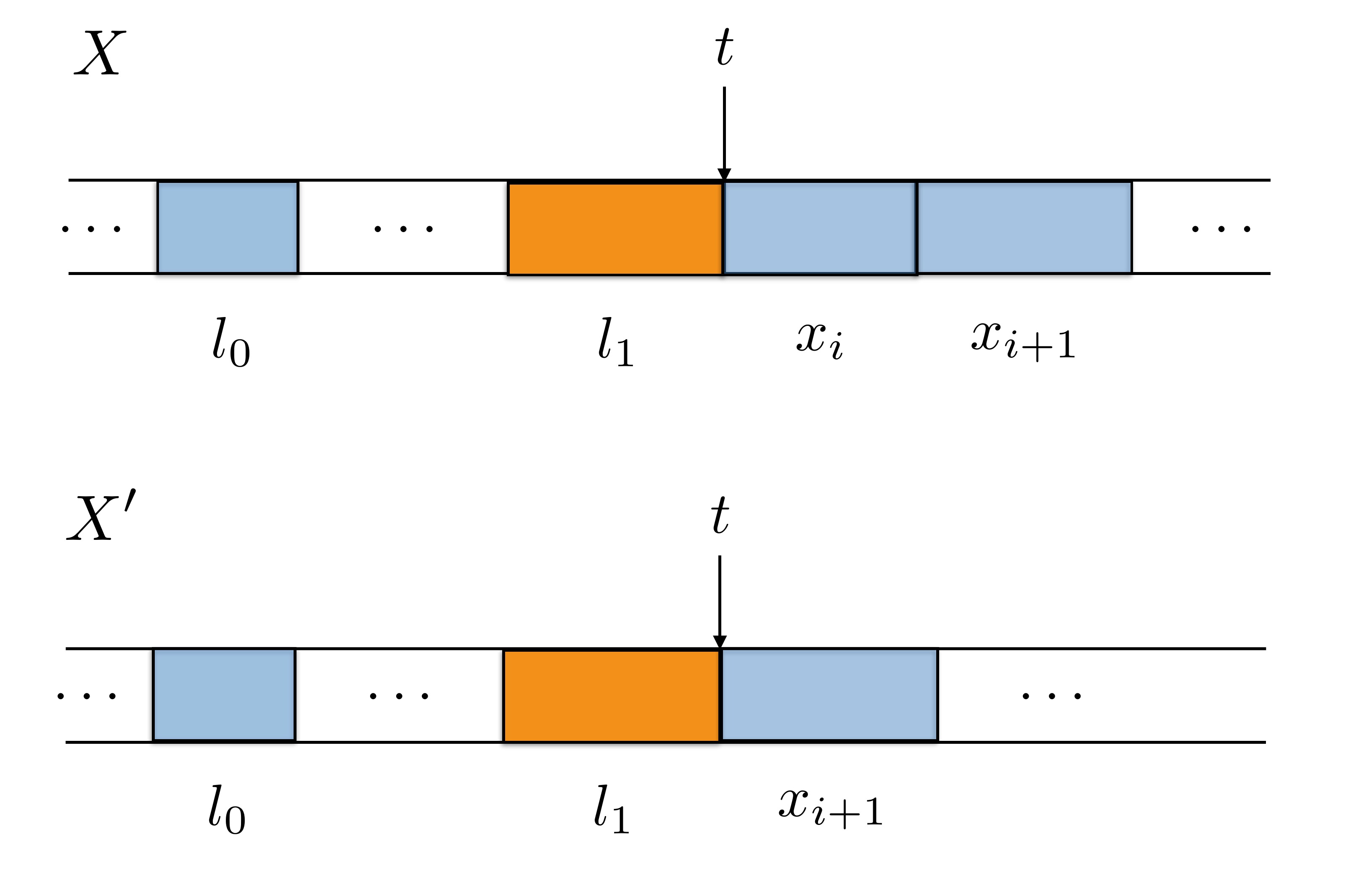}
\caption{An illustration of schedules $X$ (top) and $X'$ (bottom). $X'$ is derived by effectively removing $x_i$ from $X$.}
\label{fig:normalization}
\end{figure}


We thus must compare the deficiency $X$ to that of $X'$. To this end, note that for $X$ one needs only to consider interruptions that occur at $t^-$, as well as right before contracts $x_i,x_{i+1}$ are completed, whereas for $X'$, one needs to consider interruptions 
at time $t^-$ and right before contract $x_{i+1}$ is completed (in $X'$). All other
interruptions have the same, or better contribution to the deficiency in $X'$ than in $X$. Define $Q_X$ and $Q_{X'}$ as
\begin{eqnarray*}
Q_X&=&\max (\frac{t}{l_0+\lambda_1}, \frac{t+x_i}{l_0+l_1}, \frac{t+x_i+x_{i+1}}{x_i+l_1}) \ \textrm{ and } \nonumber \\ 
Q_{X'}&=&\max ( \frac{t}{l_0+\lambda_1}, \frac{t+x_{i+1}}{l_0+l_1}),
\end{eqnarray*}
then it follows that if $Q_{X'} \leq Q_X$ then $\defi(X') \leq \defi(X)$.

First, we observe that if $\frac{t}{l_0+\lambda_1} \geq \frac{t+x_{i+1}}{l_0+l_1}$, then  $Q_{X'} \leq Q_X$ and thus $\defi(X') \leq \defi(X)$. Thus we may assume that
$\frac{t}{l_0+\lambda_1} < \frac{t+x_{i+1}}{l_0+l_1}$. Combined with $\lambda_1\leq l_0$, this implies that 
$t<\frac{2l_0}{l_1-l_0}x_{i+1}$, therefore we obtain
\begin{equation}
t+x_{i+1} < \frac{l_0+l_1}{l_1-l_0} x_{i+1}.
\label{eq:2.problems.1}
\end{equation}
We also observe that if $\frac{t+x_i+x_{i+1}}{x_i+l_1} \geq \frac{t+x_{i+1}}{l_0+l_1}$, then again $Q_{X'} \leq Q_X$ 
and thus $\defi(X') \leq \defi(X)$.
Thus we may also assume that $\frac{t+x_i+x_{i+1}}{x_i+l_1} < \frac{t+x_{i+1}}{l_0+l_1}$, from which we obtain that 
\begin{equation}
t+x_{i+1} > \frac{x_i(l_0+l_1)}{x_i-l_0}>\frac{l_1(l_0+l_1)}{l_1-l_0}.
\label{eq:2.problems.2}
\end{equation}
From~\eqref{eq:2.problems.1} and~\eqref{eq:2.problems.2} it follows that
\[
\frac{l_1(l_0+l_1)}{l_1-l_0} <t+x_{i+1} < \frac{l_0+l_1}{l_1-l_0} x_{i+1},
\]
which in turn implies that $x_{i+1} > l_1$. This concludes the proof.  
\end{proof}

We will now use the property of the optimal schedules shown in Lemma~\ref{lemma:two.problems} in combination with Theorem~\ref{thm:limit} so as to obtain an improved lower bound on the deficiency of any schedule. 

\begin{theorem}
For $n=2$ and $m=1$, any schedule has deficiency at least 2.115.
\label{thm:exp.optimal.n.equals.2}
\end{theorem}
\begin{proof}
From Lemma~\ref{lemma:two.problems}, there exists a schedule $X$ of optimal deficiency such that at most two consecutive contracts are executed for each problem. Let $X=(x_0,x_1,\ldots)$ denote such a schedule. Consider an interruption right before contract $x_{k+1}$ is completed, and suppose, without loss of generality, that the said contract is for problem 0. Then, from Lemma~\ref{lemma:two.problems}, it follows that either $x_k$ and $x_{k-1}$
are the largest contracts for problems $0$ and $1$, respectively, that are completed by the interruption, or $x_k$ and $x_{k-2}$ are the largest such contracts for problems $1$ and $0$, respectively. Therefore, the deficiency of $X$ is at least
\[
\defi(X) \geq\sup_{k \geq 0} \frac{\sum_{j=0}^{k+1}x_j}{x_k+x_{k-1}+x_{k-2}}.
\]
Define now the functional $F_k(X)=\frac{\sum_{j=0}^{k+1}x_j}{x_k+x_{k-1}+x_{k-2}}$. Once again, it is easy to verify that 
$F_k(X)$ satisfies the conditions of Theorem~\ref{thm:limit}. Therefore, we conclude that there is $a>0$ such that 
\begin{equation}
\defi(X) \geq \sup_{k \geq 0} 
\frac{\sum_{j=0}^{k+1}a^j}{a^k+a^{k-1}+a^{k-2}}, \ \textrm{for some } a\geq 0.
\label{eq:n.1.a}
\end{equation}
Similar to the proof of Theorem~\ref{thm:roundrobin}, it is easy to see that if $a<1$, then the LHS of~\eqref{eq:n.1.a} is
$\infty$. We can thus assume that $a>1$, in which case we obtain that
\begin{eqnarray*}
\defi(X) &\geq&  \sup_{k \geq 0} \frac{\sum_{j=0}^{k+1}a^j}{a^k+a^{k-1}+a^{k-2}} \nonumber \\
&=&\sup_{k \geq 0} \frac{a^4-1/a^{k-2}}{(a-1)(1+a+a^2)}=
\frac{a^4}{a^3-1}.
\nonumber 
\end{eqnarray*}
Last, the expression $\frac{a^4}{a^3-1}$ attains its minimum at $a=2^{2/3}$. For this value of $a$, we have that 
$\defi(X) \geq 2.115$.
\end{proof}

\section{Conclusion}
\label{sec:conclusions}

In this paper we introduced and studied the problem of devising 
interruptible algorithms by means of schedules of contract algorithms, 
in a setting in which solutions
to all problem instances are required and all problems are equally important.
This generalizes the case where only one problem is queried at interruption time, and for which the standard performance measure is
the acceleration ratio. 
We introduced the deficiency of a schedule as a measure that reflects performance
in a more faithful manner than either the acceleration 
ratio or the performance ratio. We next presented a schedule whose deficiency we showed is bounded by
a small constant (at most 3.74 if $n\leq m$ and at most 4, if $n>m$; numerical evaluation provides even smaller values). 
Furthermore, for the case of one processor, we presented several lower bounds
on the deficiency of a schedule in a variety of settings, assuming a single processor. 

Even though our proposed solutions are efficient, it would be nevertheless interesting to know 
whether our schedule is theoretically optimal for any number of processors. This appears to be a challenging task. Even for the 
case of a single processor ($m=1$) and $n$ problems, the deficiency depends on the sum of $n$ contract lengths, and this
quantity appears in the denominator of the fractions that describe the deficiency~\eqref{eq:def}. This creates several technical
complications that make the application of tools such as Theorem~\ref{thm:limit} quite difficult. As shown in Lemma~\ref{lemma:two.problems},
it is possible that successive contracts for the same problem may be executed, which implies that Theorem~\ref{thm:limit} cannot 
yield an optimal lower bound, unless one derives further structural properties of the optimal schedule.

The situation becomes even more complicated for general $m$, since minimizing makespan is NP-hard. 
More importantly, one needs an upper bound on the makespan of a schedule (i.e., the output of a makespan-minimization algorithm)
in a closed-form expression as a function of the sizes of jobs. Graham's greedy algorithm yields such a closed-form expression, 
but it is no better than a 2-approximation. Thus, unlike the case $m=1$, it is not clear
how to express the exact deficiency even of simple schedules such as the exponential schedule. 

Another direction for further research is to consider the generalization in
which at interruption time a subset of $i \leq n$ problems are queried for their solution. 
We are also interested in a related setting in the context of robot searching in multiple concurrent rays. 
\cite{BFZ} and~\cite{spyros:rays} showed interesting 
connections between ray searching for a single target under the competitive ratio and contract scheduling under the acceleration ratio.  
Suppose, however, that instead of a single target, $i$ targets must be located, at unknown positions 
from the common origin. Suppose we have $m$ robots available to explore the rays. What is the best strategy
for the robots so as to locate the targets? There are interesting parallels between the two problems,
e.g., the optimal solution in the multi-target variant can be formulated as a makespan scheduling
problem, with job sizes a function of the distances of the targets from the origin. 
We believe the approach in this paper could be useful in addressing this problem.

\bibliographystyle{plain} 
\bibliography{deficiency}

\end{document}